%% file: reseff.tex
\documentclass[a4paper,11pt,twoside]{article}

\usepackage[affil-it]{authblk}
\usepackage[a4paper, inner=2.6cm, outer=2.6cm, top=3.3cm, bottom=4.2cm]{geometry}
\usepackage{fancyhdr}
\usepackage{garamond}
\usepackage{apalike}
\usepackage{amsthm}
\usepackage{amsmath}
\usepackage{algorithm}
\usepackage{algpseudocode}
\usepackage{graphicx}
\usepackage{color}
\usepackage{transparent}
\renewcommand{\thispagestyle}[1]{}
\author{Tong-Wook Shinn}
\affil{Department of Computer Science and Software Engineering\\
University of Canterbury\\
Christchurch, New Zealand}
\author{Tadao Takaoka}
\affil{Department of Computer Science and Software Engineering\\
University of Canterbury\\
Christchurch, New Zealand}

\title{Efficient Graph Algorithms for Network Analysis}

\newtheorem{definition}{Definition}[section]
\newtheorem{lemma}{Lemma}[section]
\newtheorem{theorem}{Theorem}[section]
\newtheorem{example}{Example}[section]

\begin{document}

\thispagestyle{fancy}
\maketitle

\subsection*{Keywords (required)}
Network analysis, Facility location, Bottleneck paths, Shortest paths, Fast algorithms

\section{Introduction}
\label{sec:intro}
We provide efficient algorithms for two major problems in network analysis. One is the Graph Center (GC) problem and the other is the Graph Bottleneck (GB) problem. The GC problem is relevant for facility location, while the GB problem is relevant for transportation and logistics.

The GC problem is to identify a pre-determined number of center vertices such that the distances or costs from (or to) the centers to (or from) other vertices is minimized. Let us take an example of dispatching fire engines from fire stations. The distance here is defined by the shortest distance from the closest station to a house. The problem is then to determine the center vertices (i.e. fire stations) so that the maximum distance to each house is minimized. In the case of hospitals the distance is defined by that from a house to the hospital. In the case of renewable resource centers, we may need to consider both distances, ``to'' and ``from''.

The bottleneck of a path is the minimum capacity of edges on the path. The Bottleneck Paths (BP) problem is to compute the paths that give us the maximum bottleneck values between pairs of vertices. The Graph Bottleneck (GB) problem is to find the minimum bottleneck value out of bottleneck paths for all possible pairs of vertices.

We give two similar algorithms that are based on binary search to solve the 1-center GC problem and the GB problem on directed graphs with unit edge costs. We achieve $\tilde{O}(n^{2.373})$\footnote{$\tilde{O}$ is a notation used to omit all polylog factors from the asymptotic time complexity.} worst case time complexity for both the 1-center GC problem and the GB problem, where $n$ is the number of vertices in the graph. This is better than the straightforward methods of solving the two problems in $O(n^{2.575})$ (\cite{Zwick}) and $O(n^{2.688})$ (\cite{DP}) time bounds, respectively. Note that the 2-center GC problem is investigated by \cite{Takaoka}.

We then combine the Bottleneck Paths (BP) problem with the well known Shortest Paths (SP) problem to compute the shortest paths for all possible flow values. We call this problem the Shortest Paths for All Flows (SP-AF) problem. We show that if the flow demand is uncertain, but between two consecutive capacity values, the unique shortest path can be computed to push that flow. If the uncertainty stretches over two intervals, we need to prepare two shortest paths to accommodate the uncertainty, etc. In introducing this new problem, we define a new semi-ring called the distance/flow semi-ring, and show that the well known algorithm by \cite{Floyd} can be used over the distance/flow semi-ring to solve the All Pairs Shortest Paths for All Flows (APSP-AF) problem. Further discussions of the SP-AF problem can be found in the paper by \cite{ST}.

\section{Preliminaries}
Let $G=\{V,E\}$ be a directed graph where $V$ is the set of vertices and $E$ is the set of edges. Let $|V| = n$ and $|E| = m$. We assume that the vertices are numbered from $1$ to $n$. Let $(i,j) \in E$ denote the edge from vertex $i$ to vertex $j$. Let $cost(i,j)$ and $cap(i,j)$ be the cost and capacity of the edge $(i,j)$, respectively. For the GC and GB problem in Sections \ref{sec:gc} and \ref{sec:gb}, respectively, we deal with graphs with unit edge costs, hence $cost(i,j) = 1$ for all $(i,j) \in E$. $cap(i,j)$ can be any non negative real number. Let $c$ be the maximum value of $cap(i,j)$.

Let $Z=X \star Y$ denote the Boolean matrix multiplication of matrices $X = \{x_{ij}\}$ and $Y = \{y_{ij}\}$, where $Z=\{z_{ij}\}$ is given by:
$$
z_{ij} = \bigvee\limits_{k=1}^{n} \{x_{ik} \land y_{kj}\}
$$

\begin{figure}
\def\svgwidth{270pt}
\begin{center}
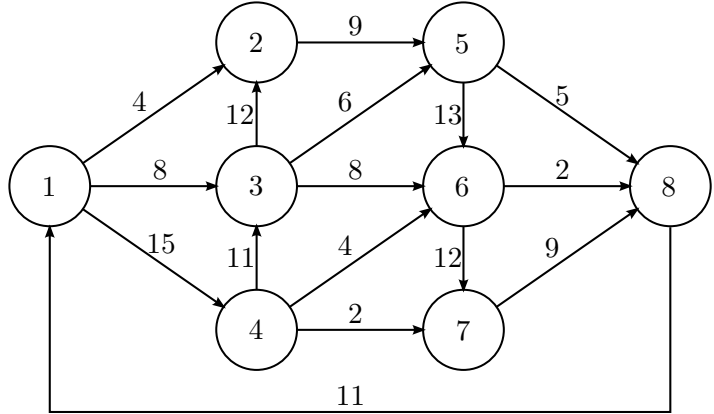
\end{center}
\caption{An example of a strongly connected directed graph with unit edge costs, with $n=8$ and $m=16$. Capacities are shown beside each edge.}
\label{fig:graph}
\end{figure}

\section{The 1-center GC problem}
\label{sec:gc}
The GC problem is closely related to the All Pairs Shortest Paths (APSP) problem. Let us assume that the APSP problem has been solved for the given graph, that is, the matrix $D^{*}$ has been solved with the shortest distance from vertex $i$ to vertex $j$ being $d^{*}_{ij}$. Then the 1-center is given by vertex $i$ that gives the minimum $\Delta$ in the following equation, where the value of $\Delta$ is the distance from the center to the farthest vertex:

$$
\Delta = \min\limits_{i=1}^{n} \{\max\limits_{j=1}^{n} d^{*}_{ij}\}
$$

The aim of our algorithm is to compute the center without computing $D^{*}$, that is, without solving the APSP problem, which is costly. As mentioned in Section \ref{sec:intro}, the current best time bound for solving the APSP problem is $O(n^{2.575})$ by \cite{Zwick}.

Let the threshold value $t$ be initialized to $n/2$. For simplicity we assume $n$ is a power of 2. Let a Boolean matrix $B$ be defined by its element $b_{ij}$ as follows: $b_{ij} = 1$ if there is an edge $(i,j)$, and 0 otherwise. We let $b_{ii}=1$ for all $i$. Let $B^{\ell}$ be the $\ell^{th}$ power of $B$ under Boolean matrix multiplication. We observe that for the matrix $B^{\ell}$, $b^{\ell}_{ij}$ = 1 if and only if $j$ is reachable from $i$ via a path whose path length is at most $\ell$. Let $C = B^{\ell}$. Let us compute:

$$
P(C) = \bigvee\limits_{i=1}^{n} \{\bigwedge\limits_{j=1}^{n} c_{ij}\}
$$

\noindent
from which we can derive the fact that $\Delta \leq \ell$ if and only if $P(C) = 1$. We can repeatedly halve the possible range $[\alpha,\beta]$ for $\Delta$ by adjusting the threshold value of $t$ through the binary search. Algorithm \ref{alg:gc} solves the GC problem in $\tilde{O}(n^{2.373})$ worst case time complexity.

\begin{algorithm}
\caption{Solve the 1-center GC problem}
\label{alg:gc}
\begin{algorithmic}[1]
\algnotext{EndFor}
\algnotext{EndIf}
\algnotext{EndWhile}
\algnotext{EndProcedure}
\State{Compute $B^{2}$, $B^{4}$, ..., $B^{n/2}$ by repeated squaring}\label{line:square}
\State{$\alpha \leftarrow 0$; $\beta \leftarrow n$}
\State{$C \leftarrow I$ /* $I$ is the unit Boolean matrix */}
\While{$\beta - \alpha > 1$}\label{line:while}
	\State{$t \leftarrow (\alpha + \beta)/2$}
	\State{$r \leftarrow (\beta - \alpha)/2$}
	\If{\Call{P}{$C \star B^{r}$} = 1}
		\State{$\beta \leftarrow t$}
	\Else
		\State{$\alpha \leftarrow t$}
		\State{$C \leftarrow C \star B^{r}$}
	\EndIf
\EndWhile
\State{$C \leftarrow C \star B^{r}$}
\State{Find row $i$ such that $c_{ij} = 1$ for all $j$ /* Vertex $i$ is the graph center */}\label{line:center}

\item[]

\Procedure{P}{$C$}
	\State{\Return{$\bigvee\limits_{i=1}^{n} \{\bigwedge\limits_{j=1}^{n} c_{ij}\}$}}
\EndProcedure
\end{algorithmic}
\end{algorithm}

\begin{lemma}
\label{lem:gc}
At the beginning of each iteration of the while-loop at line \ref{line:while} of Algorithm \ref{alg:gc}, $\alpha < \Delta \leq \beta$.
\end{lemma}
\begin{proof}
Proof is based on induction on the repetition of the while-loop. We first prove that at the beginning of each iteration of the while-loop, $C = B^{\alpha}$. At the $0^{th}$ repetition, $C = B^{0}$ and $\alpha = 0$. Suppose the lemma is true for $\alpha$. If $P(C \star B^{r}) = 1$, $C$ and $\alpha$ are unchanged. If $P(C \star B^{r}) = 0$, $\alpha$ becomes $t$ and $C$ becomes $B^{t}$. Now the lemma is true after the initialization of $\alpha$ and $\beta$. Suppose the lemma is true at the beginning of an iteration. If $P(C \star B^{r}) = 1$, $\alpha < \Delta \leq \alpha + r = t$ and $\beta$ is set to $t$. If $P(C \star B^{r}) = 0$, $t < \Delta \leq \beta$ and $\alpha$ is set to $t$.
\end{proof}

\begin{theorem}
Algorithm \ref{alg:gc} computes the graph center in $\tilde{O}(n^{\omega})$ time, where $\omega = 2.373$ (\cite{Williams}).
\end{theorem}
\begin{proof}
Upon termination of the while loop, we see $\alpha + 1 = \beta$, $\Delta = \beta$ and $C = B^{\alpha}$. Thus line \ref{line:center} of the program successfully computes the graph center. The computation of line \ref{line:square} for the powers of $B$ takes $O(n^{\omega}\log{n})$, where $O(n^{\omega})$ is the time taken for multiplying two $n$-by-$n$ Boolean matrices. Obviously the iteration in the while loop is done $O(\log{n})$ times. Thus the total time is $O(n^{\omega}\log{n}) = \tilde{O}(n^{\omega})$.
\end{proof}

Note that Algorithm \ref{alg:gc} is based on the ``to'' distance. If we use the ``from'' distance, $i$ and $j$ will be swapped in $P(C)$. If we consider both distances, we take the $\land$ operation of the two formulae.

For a graph with integer edge costs bounded by $c$, we can transform $G$ to have $O(cn)$ vertices and unit edge costs such that the 1-center GC problem can be solved in $\tilde{O}((cn)^{\omega})$ time (\cite{AGM}).

\section{The GB problem}
\label{sec:gb}
Let $\Theta$ be the bottleneck value of the entire network. The straightforward method to compute $\Theta$ would be to solve the All Pairs Bottleneck Paths (APBP) problem and find the minimum among the bottleneck paths. The current best time bound for solving the APBP problem is $\tilde{O}(n^{2.688})$ by \cite{DP}. We avoid solving the APBP problem to compute $\Theta$.

\begin{example}
\label{eg:bottleneck}
The value of $\Theta$ for the graph in Figure \ref{fig:graph} is 9, which is the capacity of edges $(2,5)$ and $(7,8)$.
\end{example}

Let the threshold value $t$ be initialized to $c/2$, where $c$ is the maximum capacity. Let Boolean matrix $B$ be defined by its element $b_{ij}$ as follows: $b_{ij} = 1$ if $cap(i,j) \geq t$, and 0 otherwise. Let us compute the transitive closure, $B^{∗}$, of $B$. Then, from the equation:

$$
b^{*}_{ij}=\Sigma \{b_{ik_{1}}b_{k_{1}k_{2}}...b_{k_{r}j}\mbox{ } | \mbox{ all possible paths } (i,k_{1}), (k_{1},k_{2}), ..., (k_{r},j)\}
$$

\noindent
we observe that $b^{∗}_{ij} = 1$ if and only if $b_{ik_{1}} = 1$, $b_{k_{1}k_{2}} = 1$, ..., $b_{k_{r}j} = 1$ for some path. From this we derive the fact that $\Theta \geq t$ if and only if $b^{*}_{ij} > 0$ for all $i$ and
$j$. We can repeatedly halve the possible range $[\alpha, \beta]$ for $\Theta$ by adjusting the threshold
value of $t$ through binary search.

\begin{algorithm}
\caption{Solve the GB problem}
\label{alg:gb}
\begin{algorithmic}[1]
\algnotext{EndFor}
\algnotext{EndIf}
\algnotext{EndWhile}
\State{$\alpha \leftarrow 0$; $\beta \leftarrow c$}
\While{$\beta - \alpha > 0$}
	\State{$t \leftarrow (\alpha + \beta)/2$}
	\For{$i \leftarrow 1$ to $n$; $j \leftarrow 1$ to $n$}
		\If{$cap(i,j) > t$}
			\State{$b_{ij} \leftarrow 1$}
		\Else
			\State{$b_{ij} \leftarrow 0$}
		\EndIf
	\EndFor
	\State{Compute $B^{*}$ /* This takes $O(n^{\omega})$ time */}\label{line:closure}
	\If{$b^{*}_{ij} > 0$ for all $i$ and $j$}
		\State{$\alpha \leftarrow t$}
	\Else
		\State{$\beta \leftarrow t$}
	\EndIf
\EndWhile
\State{$\Theta \leftarrow \alpha$}
\end{algorithmic}
\end{algorithm}

Obviously the iteration over the while-loop in Algorithm \ref{alg:gb} is performed $O(\log{c})$ times. Thus the total time complexity becomes $O(n^{\omega}\log{c})$. If $c$ is large, say $O(2^{n})$, the algorithm is not very efficient, taking $O(n)$ halvings of the possible ranges of $\Theta$. In this case, we sort edges in ascending
order of $cap(i,j)$. Since there are at most $m$ possible values of capacities, where $m$ is the number of edges,
doing binary search over the sorted edges gives us $O(n^{\omega}\log{m}) = O(n^{\omega}\log{n})$\footnote{$O(m\log{n})$ can be achieved by determining strongly connected components in $O(m)$ time (\cite{Tarjan}) rather than computing the transitive closure in each iteration.}. We note that the actual bottleneck path can be obtained with an extra polylog factor using the witness technique by \cite{AGMN}. We omit the correctness proof of Algorithm \ref{alg:gb} since it is essentially the same as the proof of Lemma \ref{lem:gc}.

\section{Algebraic treatment for graph paths problems}
Let $S$ be a closed semi-ring, that is, $S = (S, +, \cdot, 0, 1)$. The operations $+$ and $\cdot$ are associative, $+$ is commutative, and $\cdot$ distributes over $+$. $0$ is the unit element with $+$ and $1$ is the unit element with $\cdot$. For $a \in S$, the closure of $a$, $a^{*}$, is defined by $a^{*} = 1 + a + a^{2} + ...$.

Let us define matrices over semi-ring $S$. Let $O$ and $I$ correspond to $0$ and $1$ in the semi-ring. Let $M(S)$ be the set of all matrices of size $n$-by-$n$ for a fixed $n$. Let $I$ be the identity matrix and $O$ be the zero matrix. The multiplication and addition of two matrices are defined in the conventional way. For $A = \{a_{ij}\}$, $B = \{b_{ij}\}$ and $C = \{c_{ij}\}$, let C = A + B. Then $c_{ij}$ is defined by $c_{ij} = a_{ij} + b_{ij}$. Let $C = A \cdot B$. Then $c_{ij}$ is defined by:
$$
c_{ij} = \sum\limits_{k=1}^n \{a_{ik} \cdot b_{kj}\}
$$
\noindent
and the system $M(S) = (M(S), +, \cdot, O, I)$ becomes a closed semiring.

Let $R_{1}$ be the set of non-negative real numbers. $R_{1}$ is intended to represent edge costs. Let $min$ and $+$ on $R_{1}$ correspond to $+$ and $\cdot$ of the semi-ring. Then $R_{1} = (R_{1}, min, +, \infty, 0)$ becomes a closed semi-ring, called the distance semi-ring. The system $M(R_{1}) = (M(R_{1}), +, \cdot, O, I)$ becomes a closed semi-ring, where addition is the component-wise addition and for $C = A \cdot B$, $c_{ij}$ is defined by:
$$
c_{ij} = \min\limits_{k=1}^n \{a_{ik}+ b_{kj}\}
$$
\noindent
The meaning of $A^{\ell}$ is to give for the $(i,j)$ element the shortest distance from vertex $i$ to vertex $j$ that uses $\ell$ edges. Thus the closure $A^{*}$ gives the shortest distances for all pairs of vertices. As the shortest distance from any $i$ to any $j$ can be determined by the paths of at most $n-1$ edges, we have $A^{*} = I + A + A^{2} + ... = I + A + A^{2} ... + A^{n-1}$. $A^{*}$ is the solution to the All Pairs Shortest Paths (APSP) problem.

Let $R_{2}$ be the set of non-negative real numbers. $R_{2}$ is intended to represent edge capacities. Let $max$ and $min$ on $R_{2}$ correspond to $+$ and $\cdot$ of the semi-ring, called the max-min semi-ring. Then $R_{2} = (R_{2}, max, min, 0, \infty)$ becomes a closed semi-ring. The system $M(R_{2}) = (M(R_{2}), +, \cdot, O, I)$ becomes a closed semi-ring, where addition is the component-wise addition and for $C = A \cdot B$, $c_{ij}$ is defined by:
$$
c_{ij} = \max\limits_{k=1}^n \{\min{\{a_{ik},b_{kj}\}}\}
$$
\noindent
The meaning of $A^{\ell}$ is to give for the $(i,j)$ element the maximum bottleneck values of all paths from vertex $i$ to vertex $j$ that uses $\ell$ edges. Thus the closure $A^{∗}$ gives the bottleneck values for all pairs of vertices. As the bottleneck value from any $i$ to any $j$ can be determined by the paths of at most $n-1$ edges, we have $A^{*} = I + A + A^{2} + ... = I + A + A^{2} + ... + A^{n-1}$. $A^{*}$ is the solution to the All Pairs Bottleneck Paths (APBP) problem.

\section{The distance/flow semi-ring}
\label{sec:dfsr}
So far we have defined the distance semi-ring and max-min semi-ring independently. In this section we combine them and make a composite semi-ring called the distance/flow semi-ring. We define the $df$-pair, $(d,f)$, where $d$ and $f$ are from $R_{1}$ and $R_{2}$, respectively. That is, $d$ represents distance and $f$ represents flow. We define two orders on $df$-pairs.

\begin{definition}
Let $(d,f)$ and $(d',f')$ be two $df$-pairs. Then the merit order $\leq_{m}$ and the natural order $\leq_{n}$ are defined as:
$$
\begin{array}{rcl}
(d,f) \geq_{m} (d',f') & \Leftrightarrow & d \leq d' \land f \geq f' \\
(d,f) \leq_{n} (d',f') & \Leftrightarrow & d \leq d' \land f \leq f'
\end{array}
$$
\end{definition}

The meaning of the merit order is that a $df$-pair is more desirable if it has a higher $f$ for the same or lower value of $d$. If $(d,f) <_{n} (d',f')$, the $df$-pairs are incomparable under the merit order. Note that these two orders are partial orders on the direct product of integer sets.

\begin{definition}
The addition and multiplication on $df$-pairs $(d,f)$ and $(d',f')$ are defined as:
$$
\begin{array}{rcl}
(d,f) + (d',f') & = & (d,f), \text{if\ } (d,f) >_{m} (d',f') \\
	& = & (d',f'), \text{if\ } (d,f) <_{m} (d',f') \\
	& = & \{(d,f),(d',f')\}, \text{if\ } (d,f) <_{n} (d',f') \\
	& = & \{(d',f'),(d,f)\}, \text{if\ } (d,f) >_{n} (d',f') \\
(d,f) \cdot (d',f') & = & (d + d',\min{\{f,f'\}})
\end{array}
$$
\end{definition}

Note that the addition of $df$-pairs can result in a set of $df$-pairs that are incomparable under the merit order, and sorted in natural order. Intuitively speaking, addition of two $df$-pairs is to take a better route from two parallel connections from a vertex to another vertex and take both if they are incomparable. Multiplication is to compute the new $df$-pair for a serial connection. We sometimes omit the signs ``$\{  \}$'' for singletons.

The domain $R_{1} \times R_{2}$ is not closed under the $+$ operation. Thus we need to extend the domain to the power set of incomparable $df$-pairs, that is, the set of all subsets of incomparable $df$-pairs. We firstly define the $+$ operation on sets of incomparable $df$-pairs. Let $x$ and $y$ be sets of incomparable $df$-pairs, sorted in natural order. Let $z = x + y$, where $z$ can be computed as follows: We start with $x \cup y$, and remove all $df$-pairs from the union that are smaller than any other $df$-pair in the merit order. If multiples of equal $df$-pairs exist we remove all but one. Then $z$ is the resulting set of incomparable $df$-pairs, sorted in natural order. $x+y$ can be calculated in $O(|x| + |y|)$ time with Algorithm \ref{alg:add}. In this algorithm, the operation $a \Leftarrow x$ means that the element $a$ is removed from the set $x$, where $a$ is the first $df$-pair and $x$ is the set of incomparable $df$-pairs sorted in natural order. If $x$ is an empty set, that is $x = \phi$, $a \Leftarrow x$ results in $a = null$.

\begin{algorithm}
\caption{Add two sets of $df$-pairs}
\label{alg:add}
\begin{algorithmic}[1]
\algnotext{EndFor}
\algnotext{EndIf}
\algnotext{EndWhile}
\State{$z \leftarrow \phi$ /* $\phi$ means empty */ }
\State{$a \Leftarrow x$, $b \Leftarrow y$}
\While{$(a \neq null)$ and $(b \neq null)$}
	\If{$a$ and $b$ are incomparable}
		\If{$a <_{n} b$}\label{line:comp}
			\State{Append $a$ to $z$, $a \Leftarrow x$}\label{line:add_a}
		\Else
			\State{Append $b$ to $z$, $b \Leftarrow x$}\label{line:add_b}
		\EndIf
	\ElsIf{$a >_{m} b$}
		\State{$b \Leftarrow y$}\label{line:discard_b}
	\ElsIf{$b >_{m} a$}
		\State{$a \Leftarrow x$}\label{line:discard_a}
	\Else{ /* $a = b$ */}
		\State{Append $a$ to $z$, $a \Leftarrow x$, $b \Leftarrow y$}\label{line:same}
	\EndIf
\EndWhile
\If{$a \neq null$}
	\State{Append $a$ to $z$, Append $x$ to $z$ (if $x \neq \phi$)}
\EndIf
\If{$b \neq null$}
	\State{Append $b$ to $z$, Append $y$ to $z$ (if $y \neq \phi$)}
\EndIf
\end{algorithmic}
\end{algorithm}

\begin{theorem}
Algorithm \ref{alg:add} computes $x+y$ in $O(|x|+|y|)$ time, where $x$ and $y$ are both sets of incomparable $df$-pairs sorted in natural order.
\end{theorem}
\begin{proof}
We prove that the set of $df$-pairs, $z$, accumulates incomparable $df$-pairs in natural order. Observe that $df$-pairs lower in merit order are discarded at lines \ref{line:discard_b} and \ref{line:discard_a}. We discard $df$-pairs until $a$ and $b$ are incomparable or $a=b$, at which point we append one $df$-pair to $z$ (at lines \ref{line:add_a} or \ref{line:add_b} or \ref{line:same}). This ensures that all resulting $df$-pairs in $z$ are incomparable. When appending a $df$-pair to $z$ we ensure that the smaller $df$-pair (in natural order) is appended (line \ref{line:comp}). This ensures that all $df$-pairs in $z$ are sorted in natural order. $O(|x|+|y|)$ is obvious since at least one of $x$ or $y$ becomes shorter in each iteration.
\end{proof}

Now we define the $\cdot$ operation on sets of incomparable $df$-pairs. Let $x$ and $y$ be sets of incomparable $df$-pairs, sorted in natural order. Let $z = x \cdot y$, where $z$ can be computed as follows: Let $x \times y = \{a \cdot b | a \in x \land b \in y\}$. From $x \times y$, we extract all incomparable $df$-pairs. If duplicate $df$-pairs exist in the extracted set, we keep one and discard the rest. Then $z$ is the resulting set sorted in natural order. A straightforward method to compute $x \cdot y$ would take $O(|x| * |y|)$ time. We can reduce the time complexity to $O(|x| + |y|)$ with Algorithm \ref{alg:cdot}.

\begin{algorithm}
\caption{Multiply two sets of $df$-pairs}
\label{alg:cdot}
\begin{algorithmic}[1]
\algnotext{EndFor}
\algnotext{EndIf}
\algnotext{EndWhile}
\State{$z \leftarrow \phi$ /* $\phi$ means empty */ }
\State{$a \Leftarrow x$, $b \Leftarrow y$}
\While{$(a \neq null)$ and $(b \neq null)$}
	\State{Append $a \cdot b$ to $z$}
	\If{$a.f < b.f$}
		\State{$a \Leftarrow x$}\label{line:discard}
	\ElsIf{$b.f < a.f$}
		\State{$b \Leftarrow y$}
	\Else
		\State{$a \Leftarrow x$, $b \Leftarrow y$}
	\EndIf
\EndWhile
\end{algorithmic}
\end{algorithm}

\begin{theorem}
Algorithm \ref{alg:cdot} computes $x \cdot y$ in $O(|x| + |y|)$ time, where $x$ and $y$ are both sets of incomparable $df$-pairs sorted in natural order.
\end{theorem}
\begin{proof}
Suppose we have some accumulation of incomparable $df$-pairs in natural order. If $a.f < b.f$, $a$ can no longer combine with remaining $df$-pairs in $y$ to generate an incomparable $df$-pair, hence $a$ is discarded (line \ref{line:discard}). Similar reasoning applies to the case of $a.f > b.f$. If $a.f = b.f$, both are discarded.
\end{proof}

Now let us extend the distance/flow semi-ring to matrices. Let $M$ be the set of all possible $n$-by-$n$ matrices where each element of the matrix is a set of incomparable $df$-pairs sorted in natural order. Then the zero matrix, $O$, has $\{(\infty,0)\}$ as all its elements. The unit matrix, $I$, is defined as the matrix with $\{(0,\infty)\}$ for the diagonals and $\{(\infty,0)\}$ for all other elements. The multiplication and addition of these matrices are defined in the usual way using operations $+$ and $\cdot$, respectively, as defined earlier. Obviously $(M, +, \cdot, O, I)$ forms a closed semi-ring.

\section{The SP-AF problem}
Suppose there exists multiple parallel paths of varying distances and bottlenecks from a source vertex to a destination vertex. We may only be able to push a small amount of flow through a shorter path because the path may have a relatively small bottleneck value. A longer path may support a bigger flow. Clearly it is useful to determine all shortest paths for varying flow amounts. Let $t$ be the number of distinct edge capacities. $t = m$ if all edge capacities are distinct. We refer to the distinct edge capacities as \emph{maximal flows}. Then the Shortest Paths for All Flows (SP-AF) problem is to compute the shortest paths for all maximal flow values for pairs of vertices.

Using the distance/flow semi-ring as defined in Section \ref{sec:dfsr}, we can provide a formal definition for the SP-AF problem as the problem of computing all incomparable $df$-pairs for pairs of vertices, preferably sorted in natural order. Then each $df$-pair corresponds to a path between the vertices. The All Pairs Shortest Paths for All Flows (APSP-AF) problem is to solve the SP-AF problem for all possible pairs of vertices on the graph.

We define the $n$-by-$n$ matrix $A = \{a_{ij}\}$ by $a_{ij} = (cost(i,j),cap(i,j))$, for all vertex pairs $(i,j)$ in $V \times V$. Then the closure $A^{*} = I + A + A^{2} + ... = I + A + A^{2} + ... + A^{n-1}$ is the solution to the APSP-AF problem. If we perform repeated squaring on $(I + A)$, we can compute $A^{*}$ in $O(tn^{3}\log{n})$ time. We can solve the APSP-AF problem in $O(tn^{3})$ time by generalizing the well known APSP algorithm by \cite{Floyd}, as shown in Algorithm \ref{alg:apspaf}.

\begin{algorithm}
\caption{Solve the APSP-AF problem}
\label{alg:apspaf}
\begin{algorithmic}[1]
\algnotext{EndFor}
\algnotext{EndIf}
\algnotext{EndWhile}
\State{A = A + I}
\For{$k = 1$ to $n$}
	\For{$i = 1$ to $n$; $j = 1$ to $n$}
		\State{$a_{ij} = a_{ij} + a_{ik} \cdot a_{kj}$}\label{line:relax}
	\EndFor
\EndFor
\end{algorithmic}
\end{algorithm}

\begin{theorem}
Algorithm \ref{alg:apspaf} computes the closure $A^{*}$ in $O(tn^{3})$ time.
\end{theorem}
\begin{proof}
The time is obvious as line \ref{line:relax} takes $O(t)$ time. Correctness proof follows. We prove that the best incomparable $df$-tuples are obtained from paths that go though vertices $1, 2, ..., k - 1$ at the beginning of the $k^{th}$ iteration. The basis is $k = 1$. $A$ is initialized by $A + I$, which means $a_{ii} = \{(\infty, 0)\}$, and $a_{ij} = (cost(i,j),cap(i,j))$ for $i \neq j$. Suppose the hypothesis is correct for $k$. Then at line \ref{line:relax}, the best tuple without going through $k$ and going through $k$ are merged and the best $df$-pairs are chosen.
\end{proof}

\section{Concluding remarks}
We showed an asymptotic improvement on the time complexity of the 1-center GC problem. The center under the average distance measure is to minimize $\Delta = \min_{i=1}^{n}\sum_{j=1}^{n} d_{ij}^{*}$. Our algorithm can also be applied for this variation of the GC problem. Using a similar approach, we also improved the asymptotic time complexity of the GB problem. The key to our achievement was circumventing the computation of APSP and APBP, for the problems of GC and GB, respectively, with a clever use of the simple binary search method.

We then combined the SP problem and the BP problem to introduce the new SP-AF problem, where the distance/flow semi-ring plays an important role. We showed that it is straightforward to generalize the APSP algorithm to solve the APSP-AF problem using the sets of incomparable $df$-pairs. Improvements in time complexities for the SP-AF problems, such as the single source problem, will be on the research agenda for the future.

\end{document}

%% file: graph.pdf_tex
\begingroup%
  \makeatletter%
  \providecommand\color[2][]{%
    \errmessage{(Inkscape) Color is used for the text in Inkscape, but the package 'color.sty' is not loaded}%
    \renewcommand\color[2][]{}%
  }%
  \providecommand\transparent[1]{%
    \errmessage{(Inkscape) Transparency is used (non-zero) for the text in Inkscape, but the package 'transparent.sty' is not loaded}%
    \renewcommand\transparent[1]{}%
  }%
  \providecommand\rotatebox[2]{#2}%
  \ifx\svgwidth\undefined%
    \setlength{\unitlength}{278.95429677bp}%
    \ifx\svgscale\undefined%
      \relax%
    \else%
      \setlength{\unitlength}{\unitlength * \real{\svgscale}}%
    \fi%
  \else%
    \setlength{\unitlength}{\svgwidth}%
  \fi%
  \global\let\svgwidth\undefined%
  \global\let\svgscale\undefined%
  \makeatother%
  \begin{picture}(1,0.57500548)%
    \put(0,0){\includegraphics[width=\unitlength]{graph.pdf}}%
    \put(0.04638417,0.30326463){\color[rgb]{0,0,0}\makebox(0,0)[lb]{\smash{$1$}}}%
    \put(0.33356859,0.10323607){\color[rgb]{0,0,0}\makebox(0,0)[lb]{\smash{$4$}}}%
    \put(0.33356859,0.3039858){\color[rgb]{0,0,0}\makebox(0,0)[lb]{\smash{$3$}}}%
    \put(0.61891999,0.30255187){\color[rgb]{0,0,0}\makebox(0,0)[lb]{\smash{$6$}}}%
    \put(0.33356859,0.50473553){\color[rgb]{0,0,0}\makebox(0,0)[lb]{\smash{$2$}}}%
    \put(0.62035392,0.10180214){\color[rgb]{0,0,0}\makebox(0,0)[lb]{\smash{$7$}}}%
    \put(0.61891999,0.5033016){\color[rgb]{0,0,0}\makebox(0,0)[lb]{\smash{$5$}}}%
    \put(0.90570532,0.30255187){\color[rgb]{0,0,0}\makebox(0,0)[lb]{\smash{$8$}}}%
    \put(0.47122554,0.12617889){\color[rgb]{0,0,0}\makebox(0,0)[lb]{\smash{$2$}}}%
    \put(0.75801087,0.32692862){\color[rgb]{0,0,0}\makebox(0,0)[lb]{\smash{$2$}}}%
    \put(0.45401843,0.01146476){\color[rgb]{0,0,0}\makebox(0,0)[lb]{\smash{$11$}}}%
    \put(0.17153488,0.42156778){\color[rgb]{0,0,0}\makebox(0,0)[lb]{\smash{$4$}}}%
    \put(0.45688628,0.22081805){\color[rgb]{0,0,0}\makebox(0,0)[lb]{\smash{$4$}}}%
    \put(0.75801087,0.43017134){\color[rgb]{0,0,0}\makebox(0,0)[lb]{\smash{$5$}}}%
    \put(0.45688628,0.42156778){\color[rgb]{0,0,0}\makebox(0,0)[lb]{\smash{$6$}}}%
    \put(0.47122554,0.52767836){\color[rgb]{0,0,0}\makebox(0,0)[lb]{\smash{$9$}}}%
    \put(0.58880753,0.20361093){\color[rgb]{0,0,0}\makebox(0,0)[lb]{\smash{$12$}}}%
    \put(0.3020222,0.20361093){\color[rgb]{0,0,0}\makebox(0,0)[lb]{\smash{$11$}}}%
    \put(0.47122554,0.32692862){\color[rgb]{0,0,0}\makebox(0,0)[lb]{\smash{$8$}}}%
    \put(0.20021341,0.32692862){\color[rgb]{0,0,0}\makebox(0,0)[lb]{\smash{$8$}}}%
    \put(0.30058827,0.40436066){\color[rgb]{0,0,0}\makebox(0,0)[lb]{\smash{$12$}}}%
    \put(0.58880753,0.40436066){\color[rgb]{0,0,0}\makebox(0,0)[lb]{\smash{$13$}}}%
    \put(0.19160985,0.22368591){\color[rgb]{0,0,0}\makebox(0,0)[lb]{\smash{$15$}}}%
    \put(0.74367161,0.2179502){\color[rgb]{0,0,0}\makebox(0,0)[lb]{\smash{$9$}}}%
  \end{picture}%
\endgroup%